\documentclass[article]{IEEEtran}
\usepackage{amsmath, amssymb, epsfig}
\usepackage{calc,pgf,pstricks,multido,cite}

\def\pp{P_{\text{alert}}}
\def\pa{P_{\text{avg}}}

\def\msg{\text{MSG}}
\def\fa{\text{FA}}
\def\md{\text{MD}}

\def \E{\mathbb E}
\def \P{\mathbb P}

\def\ea{\epsilon}
\def\eb{\lambda}
\def\ec{\delta}
\def\ed{\nu}
\def\ee{\beta}
\def\ef{\gamma}

\def\ab{\mathbf{a}}

\def\zb{\mathbf{z}}

\def\Vm{\mathcal{V}}
\def\Sm{\mathcal{S}}

\def\V{\mathsf{Vol}}

\def\n{^{-1}}

\newtheorem{theorem}{Theorem}

\newtheorem{lemma}{Lemma}

\newtheorem{definition}{Definition}

\newtheorem{remark}{Remark}

\title{The AWGN Red Alert Problem}
\author{Bobak Nazer, \IEEEmembership{Member, IEEE}, Yanina Shkel, \IEEEmembership{Graduate Student Member, IEEE}, and Stark C. Draper, \IEEEmembership{Member, IEEE}\thanks{B. Nazer, Y. Shkel, and S. C. Draper were supported by the Air Force Office of Scientific Research under Grant FA9550-09-1-0140 and by the National Science Foundation under grant CAREER 0844539. The material in this paper was presented in part at the 48th Annual Allerton Conference on Communication, Control, and Computing, Monticello, IL, September 2010 and at the 6th Annual Information Theory and Applications Workshop, UCSD, La Jolla, CA, February 2011.}\thanks{B. Nazer was with the Department of Electrical and Computer Engineering, University of Wisconsin, Madison, WI, 53706, USA and is now with the Department of Electrical and Computer Engineering, Boston University, Boston, MA, 02215 (email: \texttt{bobak@bu.edu}). Y. Shkel and S. C. Draper are with the Department of Electrical and Computer Engineering, University of Wisconsin, Madison, WI, 53706, USA; (\texttt{yyshkel@wisc.edu}; \texttt{sdraper@ece.wisc.edu})}}

\date{\today}
\begin{document}
\maketitle

\begin{abstract}
Consider the following unequal error protection scenario. One special message, dubbed the ``red alert'' message, is required to have an extremely small probability of missed detection. The remainder of the messages must keep their average probability of error and probability of false alarm below a certain threshold. The goal then is to design a codebook that maximizes the error exponent of the red alert message while ensuring that the average probability of error and probability of false alarm go to zero as the blocklength goes to infinity. This red alert exponent has previously been characterized for discrete memoryless channels. This paper completely characterizes the optimal red alert exponent for additive white Gaussian noise channels with block power constraints.
\end{abstract}

\section{Introduction}

Communication networks are increasingly being taxed by the enormous demand for instantly available, streaming multimedia. Ideally, we would like to maximize the reliability and data rate of a system while simultaneously minimizing the delay. Yet, in the classical fixed blocklength setting, the reliability function of a code goes to zero as the rate approaches capacity even in the presence of feedback. This seems to imply that, close to capacity, it is impossible to keep delay low and reliability high. However, this lesson is partially an artifact of the block coding framework. The achievable tradeoff changes in a streaming setting where all bits do not need to be decoded by a fixed deadline, but rather, each individual bit must be recovered after a certain delay. In this setting, the reliability function measures how quickly the error probability on each bit estimate decays as a function of the delay. Surprisingly, the achievable error exponent can be quite large at capacity if a noiseless feedback link is available and cleverly exploited \cite{kudryashov79,sahai08}.
 
The distinguishing feature of these streaming architectures with feedback is the use of an ultra-reliable special codeword that is transmitted to notify the decoder when it is about to make an error. While this ``red alert'' codeword requires a significant fraction of the decoding space to attain its very large error exponent, the remaining ``standard'' codewords merely need their error probability to vanish in the blocklength. One question that seems intimately connected to the streaming delay-reliability tradeoff is how large the red alert error exponent can be made for a fixed blocklength codebook of a given rate.  Beyond this streaming motivation, the red alert problem is also connected to certain sensor network scenarios. For example, consider a setting where sensors must send regular updates to the basestation using as little power as possible, i.e., using the standard codewords. If an anomaly is detected, the sensors are permitted to transmit at higher power in order to alert the basestation with high reliability, which corresponds to our red alert problem.

Prior work has characterized the red alert exponent for discrete memoryless channels (DMCs) \cite{bnz09,sd08,ngzc11}. In this paper, we determine the red alert exponent for point-to-point additive white Gaussian noise (AWGN) channels that operate under block power constraints on both the regular and red alert messages.  We derive matching upper and lower bounds on the red alert exponent with a focus on the resulting high-dimensional geometry of the decoding regions. Our code construction can be viewed as a generalization of that used in the discrete case. 

\subsection{Related Work}

Previous studies on protecting a special message over a DMC have relied on some variant of the following code construction. First, designate the special codeword to be the repetition of a particular input symbol. Then, generate a fixed composition codebook at the desired rate. This composition is chosen to place the ``standard'' codewords as far as possible from the special codeword (as measured by the Kullback-Leibler (KL) divergence between induced output distributions) while still allocating each codeword a decoding region large enough to ensure a vanishing probability of error. By construction, the rest of the space is given to the special codeword. Early work by Kudryashov used this strategy to achieve very high error exponents in the bit error setting under an expected delay constraint \cite{kudryashov79}. 

In \cite{bnz09}, Borade, Nakibo\u{g}lu, and Zheng study ``bit''-wise and ``message''-wise unequal error protection (UEP) problems and error exponents.  The red alert problem is a message-wise UEP problem in which one message is special and the remaining messages are standard.  While \cite{bnz09} focuses on general DMCs near capacity, Lemma 1 of that paper develops a general sharp bound on the red alert exponent for DMCs at any rate below capacity (both with and without feedback). Specializing to the exponent achieved at capacity, let $\mathcal{X}$ denote the input alphabet, $\{p_{Y|X}(\cdot|x)\}_{x \in \mathcal{X}}$ the channel transition matrix, and $p_Y^{\ast}(\cdot)$ the capacity-achieving output distribution of the DMC.  Then, the optimal red alert exponent at capacity is
\begin{align}
E_{\rm ALERT} (C) = \max_{x \in \mathcal{X}} D( p_Y^{\ast}(\cdot) \| p_{Y|X}(\cdot|x))
\end{align}
where $D(\cdot\|\cdot)$ is the KL divergence.  We also mention recent work by Nakibo\u{g}lu \textit{et al.} \cite{ngzc11,ngzc10} that considers the generalization where a strictly positive error exponent is required of the standard messages.  

For the binary symmetric channel (BSC), the optimal red alert exponent has a simple and illustrative form. This exponent can be inferred from the general expression in \cite[Lemma 1]{bnz09} or via a direct proof due to Sahai and Draper \cite{sd08} (which appeared concurrently with the conference version \cite{bnz08} of \cite{bnz09}). Let $p$ denote the crossover probability of the BSC and $q$ the probability that a symbol in the codebook is a one.  Then, the optimal red alert exponent as a function of rate $R < C$ for the BSC is
\begin{align}
E_{\rm ALERT} (R) = \max_{h_B(q \ast p) - h_B(p) \leq R} D(q \ast p \| p)
\end{align}
where $h_B(p) = -p \log p - (1-p) \log (1-p)$, $q \ast p =  p (1-q) + q (1-p)$, and $D(\tilde{p} \|p) = \tilde{p} \log \left( \frac{\tilde{p}}{p} \right) + (1-\tilde{p}) \log \left( \frac{1-\tilde{p}}{1-p} \right)$. 
 
Csisz\'{a}r studied a related problem where multiple special messages require higher reliability in \cite{csiszar80}. Upper bounds for multiple special messages with different priority levels were also developed in \cite{bnz09}. In \cite{bs09}, Borade and Sanghavi examined the red alert problem from a coding theoretic perspective. As shown by Wang \cite{wangmsthesis}, similar issues arise in certain sparse communication problems where the receiver must determine whether a codeword was sent or the transmitter was silent. 

The fundamental mechanism through which high red alert exponents are achieved is a binary hypothesis test. By designing the induced distributions at the output of the channel to be far apart as measured by KL divergence, we can distinguish whether the red alert or some standard codeword was sent. The test threshold is biased to minimize the probability of missed detection and is analyzed via an application of Stein's Lemma. This sort of biased hypothesis test occurs in numerous other communication settings with feedback, such as \cite{horstein63,forney68,burnashev76} and, as mentioned earlier, these codes are also used as a component in streaming data systems (see, for instance, \cite{kudryashov79,ds06,sahai08,sd10}). There is also a rich literature on the interplay between hypothesis testing and information theory, which we cannot do justice to here (see, for instance, \cite{blahut74,ac86,hk89}). 

\section{Problem Statement}

First, we mention some of our notational choices. We will use boldface lowercase letters to denote column vectors, $\mathbf{0}$ to denote the all zeros vector, and $\mathbf{1}$ to denote the all ones vector. Throughout the paper, the $\log$ function is taken to be the natural logarithm and rate is measured in nats instead of bits. We use $\| \mathbf{x} \|$ to denote the Euclidean norm of the vector $\mathbf{x}$.

\begin{definition}[Messages] The transmitter has a \textit{message} $w \in \{0, 1, 2,\ldots, M\}$ that it wants to convey to the receiver. One of the messages, $w=0$, is a red alert message that will be afforded extra error protection. We assume the red alert message is chosen with some probability greater than $0$ and the remaining messages are chosen with equal probability.
\end{definition}

\begin{definition}[Encoder] The \textit{encoder} $\mathcal{E}$ maps the message $w$ into a length-$n$ real-valued codeword $\mathbf{x}$ for transmission over the channel, $\mathcal{E}: \{0,1,2,\ldots, M\} \rightarrow \mathbb{R}^n$. Let $\mathbf{x}(w)$ denote the codeword used for message $w$ and let $\mathcal{C}$ denote the entire codebook, $\mathcal{C} = \{\mathbf{x} (0), \mathbf{x}(1), \ldots, \mathbf{x}(M)\}$. The codebook must satisfy both an average block power constraint across codewords,
\begin{align}
\frac{1}{M} \sum_{w = 1}^M \| \mathbf{x}(w) \|^2 &\leq n\pa  \ . \label{e:avgpower}
\end{align} In addition, the red alert codeword must satisfy a less stringent power constraint,
\begin{align}
\| \mathbf{x}(0) \| ^2 &\leq n\pp \ , \label{e:alertpower}
\end{align} for some $\pp \geq \pa$. The \textit{rate} of the codebook is 
\begin{align}
R = \frac{1}{n} \log M 
\end{align} nats per channel use.
\end{definition}
\begin{remark}
Note that our codebook average power constraint (\ref{e:avgpower}) is less strict that the usual block power constraint $\| \mathbf{x}(w) \|^2 \leq n\pa$. Our achievable scheme can be easily modified to meet this constraint using expurgation. Furthermore, our red alert power constraint (\ref{e:alertpower}) is less strict than a peak power constraint $| x_i(0) |^2 \leq \pp$ ~$\forall i$, where $x_i(0)$ denotes the $i$th symbol of the red alert codeword. Our scheme sets the red alert codeword to be $\mathbf{x}(0) = -\sqrt{\pp} \mathbf{1}$, which naturally satisfies a peak power constraint. Therefore, our main results hold under an average power constraint and peak power constraint as well.  
\end{remark}

\begin{remark}We omit the red alert codeword from the average block power constraint for the sake of simplicity. Another possibility would be to consider only an average block power constraint over both the standard and red alert codeword. This would lead to two different tensions between maximizing the red alert exponent and maximizing the rate. The first would be the allocation of the decoding regions and the second would be the allocation of power based on the probability of a red alert message. By using two separate power constraints, we can state our results in a simpler form that does not depend on the red alert probability.   \end{remark} 

\begin{definition}[Channel] The \textit{channel} outputs the transmitted vector, corrupted by independent and identically distributed (i.i.d.) Gaussian noise:
\begin{align}
\mathbf{y} = \mathbf{x} + \mathbf{z}
\end{align} where $\mathbf{z} \sim \mathcal{N}(\mathbf{0}, N \mathbf{I}^{n \times n})$ for some noise variance $N > 0$.
\end{definition}

\begin{definition}[Decoder] The signal observed by the receiver is sent into a \textit{decoder} which produces an estimate $\hat{w}$ of the transmitted message $w$, $\mathcal{D}: \mathbb{R}^n \rightarrow \{0,1,2,\ldots,M\}$.
\end{definition}

\begin{definition}[Error Probability] We are concerned with three quantities, the \textit{probability of missed detection} of the red alert message $p_{\md}$, the \textit{probability of false alarm} $p_{\fa}$,  and the \textit{average probability of error} of all other messages $p_{\msg}$:
\begin{align}
p_{\md} &= \P( \hat{w} \neq w | w = 0 )\\
p_{\fa} &= \P( \hat{w} = 0 | w \neq 0 )\\
p_{\msg} &= \frac{1}{M}\sum_{w=1}^{M} \P( \hat{w} \neq w | w \neq 0,\hat{w} \neq 0) \ . 
\end{align} 
\end{definition}

\begin{definition}[Error Exponent]
We say that a \textit{red alert exponent} of $E_{\text{ALERT}}(R)$ is achievable if for every $\epsilon > 0$ and $n$ large enough, there exists a rate $R$ encoder and a decoder such that
\begin{align}
-\frac{1}{n} \log\left(p_{\text{MD}}\right) &> E_{\text{ALERT}}(R) - \epsilon \\
p_{\text{FA}} &< \epsilon \\ 
p_{\text{MSG}} &< \epsilon \ . 
\end{align} 
\end{definition} In other words, we would like the red alert codeword to have as large an error exponent as possible while keeping the other error probabilities small. The standard codewords do not need to have a positive error exponent. Of course, the rate must be lower than the AWGN capacity, $R \leq C$, where 
\begin{align}
C \triangleq \frac{1}{2} \log{\left(1 + \frac{\pa}{N}\right)} \ .
\end{align} 

\subsection{High-Dimensional Geometry} 

We now review some basic facts of high-dimensional geometry that will be useful in our analysis.

Let $\mathcal{B}_n(\mathbf{a},r)$ denote the $n$-dimensional ball centered at $\mathbf{a} \in \mathbb{R}^n$ with radius $r > 0$. Recall that the volume of $\mathcal{B}_n(\mathbf{a},r)$ is
\begin{align}
\mathsf{Vol}\left(\mathcal{B}_n(\mathbf{a},r)\right) = \frac{r^n\pi^{n/2}}{\Gamma\left(\frac{n}{2} + 1\right)} \label{e:spherevol}
\end{align} where $\Gamma(\cdot)$ is the gamma function \cite[Ch. 1, Eq. 16]{conwaysloane}. We define $\mathcal{S}_n(\mathbf{a},r)$ to be the surface of $\mathcal{B}_n(\mathbf{a},r)$. Its surface area (or, more precisely, the $(n-1)$-dimensional volume of its surface) is 
\begin{align}
\mathsf{Vol}\left(\mathcal{S}_n(\mathbf{a},r)\right) = \frac{nr^{n-1}\pi^{n/2}}{\Gamma\left(\frac{n}{2} + 1\right)} 
\end{align} \cite[Ch. 1, Eq. 19]{conwaysloane}. The dimension of the $\V(\cdot)$ function will always be clear from the context. We also define 
\begin{align}
\mathcal{T}_n(\mathbf{a},r_1,r_2) \triangleq \big\{\mathbf{x}: r_1 \leq \| \mathbf{x} - \mathbf{a} \| \leq r_2 \big\} \end{align} to be the spherical shell centered at $\mathbf{a}$ from radius $r_1$ to $r_2$. 

The angle between two $n$-dimensional vectors $\mathbf{a}$ and $\mathbf{b}$ is 
\begin{align}
\measuredangle(\mathbf{a},\mathbf{b}) = \cos^{-1} \left( \frac{\mathbf{a}^T \mathbf{b}}{\| \mathbf{a}\| \|\mathbf{b}\|} \right)
\end{align} where $\cos^{-1}(\cdot)$ takes values between $0$ and $\pi$. Let $\mathcal{V}_n(\mathbf{a},\mathbf{b},\theta)$ denote the $n$-dimensional cone with its origin at $\mathbf{a}$, its center axis running from $\mathbf{a}$ to $\mathbf{b}$, and of half-angle $\theta$ which takes values from $0$ to $\pi/2$. The solid angle $\Omega(\theta)$ of an $n$-dimensional cone of half-angle $\theta$ is the fraction of surface area that it carves out of an $n$-dimensional sphere, 
\begin{align}
\Omega(\theta) = \frac{\V\Big( \Vm_n(\mathbf{0}, \mathbf{1},\theta) \cap \Sm_n(\mathbf{0},r) \Big)}{\V\big( \Sm_n(\mathbf{0},r)\big)} \ . 
\end{align} Note that the solid angle is the same for any sphere radius $r > 0$. 

\begin{lemma}[Shannon] \label{l:surfaceratio}The solid angle of a cone with half-angle $\theta$ satisfies
\begin{align}
\Omega(\theta) = \frac{\sin^n\theta}{\sqrt{2\pi n} \sin \theta \cos \theta } \bigg(1 + O\left(\frac{1}{n}\right)\bigg) \ . \nonumber
\end{align}
\end{lemma} See the math leading up to Equation 28 in \cite{shannon59} for a proof.
\section{Main Result}

In the binary case, the simplest characterization of the optimal codebook is a statistical one: the red alert codeword is the zero vector and the remaining codewords are of a constant composition. From one perspective, this can be visualized as placing the red alert codeword in the ``center'' of the space with the other codewords encircling it (see Figure \ref{f:topdown}). This corresponds to choosing the red alert codeword to be the all zeros (or all ones) vector. The standard codewords are generated using the distribution that maximizes the KL divergence between output distributions while still supporting a rate $R$. While this two-dimensional illustration is quite useful for understanding the binary case, it can be misleading in the Gaussian case. Specifically, it suggests that we should place the red alert codeword at the origin which turns out to be suboptimal. 

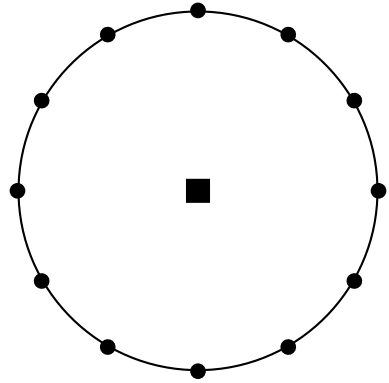
\begin{figure}[h]
\begin{center}
\psset{unit=0.8mm}
\begin{pspicture}(-30,-30)(30,30)

\multido{\n=150+30}{12}{
\rput{\n}(0,0){
\pscircle[fillstyle=solid,fillcolor=black](30,0){1.3}
}}
\pscircle[linestyle=solid](0,0){30}

\rput(0,0){\psframe[fillstyle=solid,linecolor=black,fillcolor=black](-2,-2)(2,2)}

\end{pspicture}
\end{center}
\caption{For a BSC, we can visualize the red alert codeword (solid square) sitting in the ``center'' of the codebook and the standard codewords (solid circles) occupying a thin shell around it. While this illustration is generally sufficient for developing the intuition behind the discrete case, it does not capture the full story in the Gaussian case. }\label{f:topdown}
\end{figure}

Another way of looking at the binary construction is to visualize each fixed composition as a parallel (or circle of constant latitude) on a sphere (see Figure \ref{f:sideview}). That is, the code lives on the Hamming cube in $n$ dimensions, which can be imagined as a sphere by taking the all zeros and all ones vectors as the two poles and specifying the parallels by their Hamming weight. From this viewpoint, the binary construction sets the red alert codeword to be one of the poles and chooses the remaining codewords on the furthest parallel that can support a codebook of rate $R$. This perspective leads naturally to the right construction for the Gaussian case. Essentially, the standard codewords are placed uniformly along a constant parallel. This can be achieved by generating the standard codewords using a capacity-achieving code with a fraction $\alpha$ of the total power. The red alert codeword is placed at the furthest limit of the red alert power constraint (e.g., at $-\sqrt{\pp}\mathbf{1}$) and the standard codewords are offset in the opposite direction (e.g., by $\sqrt{(1-\alpha)\pa} \mathbf{1}$). See Figure \ref{f:offsetcodebook} for an illustration. In the high-dimensional limit, most of the codewords will live on a parallel, thus mimicking the binary construction. This scheme leads us to the optimal red alert exponent.

\begin{theorem} \label{t:main} For an AWGN channel with red alert power constraint $\pp$, average power constraint $\pa$, and rate $R$, the optimal red alert exponent is
\begin{align}
E(R)= \frac{\pp + \pa + 2\sqrt{\pp(\pa + N(1 - e^{2R}))}}{2N} - R \nonumber \ . 
\end{align}
\end{theorem} We prove achievability in Lemma \ref{l:achievable} and provide a matching upper bound in Lemma \ref{l:converse}. 

In the conference version of this paper \cite{nd10allerton}, we used a different code construction that lead to a smaller achievable red alert exponent. The codewords were generated uniformly on the sphere of radius $\sqrt{n\pa}$ and we only kept those that fell within a cone of appropriate half-angle. This type of construction turns out not to achieve as dense a packing as the construction used in this paper. In Appendix \ref{a:cone}, we explore the reasons why this occurs in the binary case. In Appendix \ref{a:coneexponent}, we state the achievable red alert exponent for the conical construction.

\begin{figure}
\begin{center}
\includegraphics{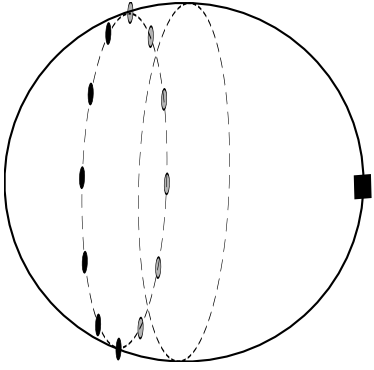}
\end{center}\caption{From an alternate viewpoint, we can visualize the red alert codeword (solid square) sitting on the pole of a sphere and the standard codewords (solid circles) on a parallel on the opposite hemisphere. Our code construction for the Gaussian case is inspired by this picture. If the red alert power constraint is larger than the average power constraint, the red alert codeword should be placed on the sphere's axis but off its the surface, directly above the pole. } \label{f:sideview}
\end{figure}

\begin{figure}[h]
\begin{center}
\psset{unit=0.85mm}
\begin{pspicture}(-50,-15)(55,15)
\psframe[fillcolor=black,fillstyle=solid](-50,-2)(-46,2)
\pscircle[fillstyle=solid,fillcolor=gray!20!white](40,0){15}
\psline[linewidth=1.5pt](-48,0)(0,0)
\rput(-24,-4){\small{$\sqrt{n\pp}$}}
\psline[linewidth=1.5pt,linestyle=dashed]{->}(0,0)(40,0)
\rput(12,-4){\small{$\sqrt{n(1-\alpha)\pa}$}}
\psline[linewidth=1.5pt]{->}(40,0)(40,15)
\rput{270}(44,3){\footnotesize{$\sqrt{n(\alpha\pa-\lambda)}$}}
\pscircle[linewidth=1.5pt,fillstyle=solid,fillcolor=white](0,0){1.3}

\end{pspicture}
\end{center}
\caption{Red alert codebook construction. The red alert codeword (solid square) is placed at $-\sqrt{\pp} \mathbf{1}$ which takes it a distance $\sqrt{n\pp}$ from the origin (circle). The standard codewords (shaded region) are drawn i.i.d. according to a Gaussian distribution with variance $\alpha \pa - \lambda$. These codewords are pushed away from the origin by an offset $\sqrt{(1-\alpha)\pa} \mathbf{1}$ (dashed line). }\label{f:offsetcodebook}
\end{figure}
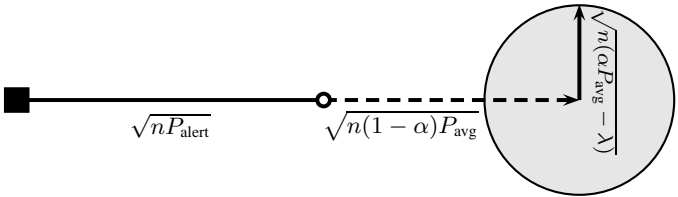

\section{Codebook Construction} \label{s:codebook}

Our codebook construction for $\mathcal{C}$ consists of the following steps:

\begin{enumerate}
\item Choose $\ea > 0$ so that $R < C - \ea$.
\item The red alert codeword is placed at the boundary of the red alert power constraint, $\mathbf{x}(0) = -\sqrt{\pp} \ \mathbf{1}$. 
\item Choose $0<\alpha \leq 1$ so that 
\begin{align}
R + \ea = \frac{1}{2}\log\left(1 + \frac{\alpha\pa }{N}\right)
\end{align} and choose $\eb > 0$ so that 
\begin{align}
R + \frac{\ea}{2} = \frac{1}{2} \log\left(1 + \frac{\alpha \pa - \eb}{N} \right) \ . 
\end{align}
\item Draw $e^{nR}$ codewords $\mathbf{v}(1),\ldots,\mathbf{v}(e^{nR})$ i.i.d. according to a Gaussian distribution with mean zero and variance $\alpha \pa - \eb$.
\item To each of these codewords, add an offset $\sqrt{(1-\alpha)\pa} \ \mathbf{1}$ so that the transmitted codeword for each message (other than $w = 0$) is $\mathbf{x}(w) = \sqrt{(1-\alpha)\pa}  \ \mathbf{1} + \mathbf{v}(w)$.
\end{enumerate}
We will show that this procedure yields a random codebook $\mathcal{C}$ whose false alarm probability and average probability of error are both less than $\epsilon$. Afterwards, we will characterize the probability of missed detection for the red alert codeword. This will in turn imply the existence of a good fixed codebook.

\section{Achievability}

In this section, we will show that the red alert error exponent stated in Theorem \ref{t:main} is achievable. We begin by stating useful large deviations bounds that will play a role in both the proof of the achievability and of the converse. Next, we show that any standard codeword plus noise lies at a certain distance from the red alert codeword with high probability. Afterwards, we argue that, with high probability, any standard codeword plus noise is contained in a cone of a certain half-angle that is centered on the red alert codeword. By combining the distance and angle bounds, we can constrain the decoding region for the standard codewords to the intersection of a cone with a shell. The remainder of $\mathbb{R}^n$ can thus be allocated to the decoding region for the red alert codeword, for which we will bound the resulting probability of a missed detection. 

\subsection{Large Deviations Bounds}

Our upper and lower bounds on the probability of error are proven by deriving bounds on the size and shape of the decoding regions and then applying Cram\'er's Theorem to get large deviations bounds. Define $g_X(a)$ to be the moment generating function of a random variable $X$,
\begin{align}
g_X(a) \triangleq \E[e^{aX}] \ , 
\end{align} and $I_X(b)$ to be the Fenchel-Legendre transform \cite[Definition 2.2.2]{dembozeitouni} of $\log(g_X(\cdot))$,
\begin{align}
I_X(b) \triangleq  \sup_a\big[ab - \log(g_X(a))\big] \ . 
\end{align}

\begin{theorem}[Cram\'er] \label{t:cramer}
Let $S_n = \frac{1}{n} \sum_i X_i$ be the normalized sum of $n$ i.i.d. variables $X_1, \ldots, X_n$ with finite mean and rate function $I_X(b)$. Then, for every closed subset $\mathcal{F} \subset \mathbb{R}$,
\begin{align}
\P(S_n \in \mathcal{F}) \leq 2 \exp\Big(-n \inf_{b \in \mathcal{F}} I_X(b)\Big) \ ,
\end{align} and, for every open subset $\mathcal{G} \subset \mathbb{R}$,
\begin{align}
\liminf_{n\rightarrow\infty}\frac{1}{n}\log\P(S_n \in \mathcal{G}) \geq  - \inf_{b \in \mathcal{G}} I_X(b) \ . 
\end{align}
\end{theorem} See, for instance, \cite[Theorem 2.2.3]{dembozeitouni} for a proof. We will be particularly interested in how this bound applies to the length of i.i.d. Gaussian vectors, which corresponds to setting the $X_i$ to be Chi-square random variables (with one degree of freedom). The moment generating function for such random variables is
$g_X(a) = \frac{1}{\sqrt{1- 2a}}$ which yields a rate function of $I_X(b) = \frac{1}{2}(b - 1 - \log b )$.
\subsection{Distance Bounds} \label{s:distance}

The following lemma formalizes the notion that the squared $\ell_2$-norm of an i.i.d. Gaussian vector concentrates sharply around its variance. Thus, for large $n$, the decoding region can be restricted to a thin spherical shell. 
\begin{lemma} \label{l:gaussiannorm}
Let $\mathbf{z}$ be a length-$n$ vector with i.i.d.~zero-mean Gaussian entries of variance $N$. Then, for any $\ee > 0$, 
\begin{align}
\P\big(\| \mathbf{z} \|^2 \geq nN(1 + \ee) \big) \leq 2 \exp\left(-\frac{n}{2}\big(\ee - \log(1 + \ee) \big)\right) \nonumber
\end{align} and, for any $0 < \ee < 1$,
\begin{align}
\P\big(\| \mathbf{z} \|^2 \leq nN(1 - \ee) \big)  \leq  2 \exp\left(-\frac{n}{2}\big(-\ee - \log(1 - \ee) \big) \right). \nonumber
\end{align} 
\end{lemma} See Appendix \ref{a:distanceproofs} for the proof.

Recall that the Q-function returns the probability that a scalar Gaussian random variable with mean zero and unit variance is greater than or equal to $t > 0$,
\begin{align}
Q(t) \triangleq \frac{1}{\sqrt{2 \pi}} \int_{t}^{\infty} \exp{\left(-\frac{t^2}{2}\right)} dx \ ,
\end{align} and is upper bounded as 
\begin{align}
Q(t) < \frac{1}{2} \exp\left(-\frac{t^2}{2}\right) \ .\nonumber 
\end{align} The next lemma is about the well-known fact that an i.i.d.~Gaussian vector is approximately orthogonal to any fixed vector.

\begin{figure}
\begin{center}
\includegraphics[width=2.2in]{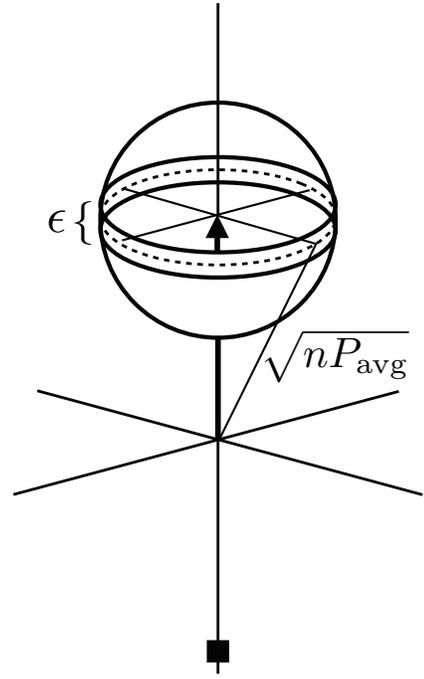}
\end{center}
\caption{From the perspective of the origin, most of the codewords are concentrated in an $\epsilon$-shell of power $\alpha \pa - \lambda$ that is offset away from the origin with power $(1-\alpha)\pa$. Thus, with high probability, any random codeword meets the power constraint.}\label{f:OffsetCB}
\end{figure}

\begin{lemma} \label{l:gaussianortho}
Let $\mathbf{z}$ be a length-$n$ vector with i.i.d.~zero-mean Gaussian entries with variance $N$ and let $\mathbf{a}$ be a length-$n$ vector with $\| \mathbf{a}\|^2 = n \alpha$ for some fixed $\alpha > 0$. Then, for any $\ec > 0$ and $n$ large enough, 
\begin{align}
\P\big( | \mathbf{a}^T \mathbf{z} | \geq \ec \| \mathbf{a} \|^2 \big) < \ec \ . 
\end{align}
\end{lemma} See Appendix \ref{a:distanceproofs} for the proof.

In Figure \ref{f:OffsetCB}, the codebook is illustrated from the perspective of the origin. Using the above lemma, it can be shown that all but a vanishing fraction of codewords have power close to $\pa$ and are nearly orthogonal with respect to any fixed vector. We now characterize how far away a codeword plus noise is from the red alert codeword with high probability.

\begin{lemma}\label{l:distance}
For any $\ec > 0$ and $n$ large enough, the distance from the red alert codeword to the codeword for a standard message, $w \in \{1,2,\ldots,e^{nR}\}$, plus noise is at least $L$ with high probability,
\begin{align}
&\P\left(\| -\mathbf{x}(0) + \mathbf{x}(w) + \mathbf{z} \| \geq L\right) > 1 - \delta \\
&L = \sqrt{n \Big( \pp + \pa + N + 2 \sqrt{\pp(1-\alpha)\pa} - \eb - \ec\Big)} \nonumber \ . 
\end{align}
\end{lemma} See Appendix \ref{a:distanceproofs} for the proof.

\subsection{Angle Bounds}\label{s:angle}

We now upper bound the $n$-dimensional angle between a fixed vector and the same vector plus i.i.d. Gaussian noise.  

\begin{lemma} \label{l:gaussianangle}
Let $\mathbf{z}$ be a length-$n$ vector with i.i.d.~zero-mean Gaussian entries with variance $N$ and let $\mathbf{a}$ be a length-$n$ vector with $\| \mathbf{a}\|^2 = n \alpha$ for some fixed $\alpha > 0$. For any $\ec > 0$ and $n$ large enough, the probability that the angle between $\mathbf{a}$ and $\mathbf{a} + \mathbf{z}$ exceeds $\cos^{-1}\left(\sqrt{\frac{\alpha}{\alpha + N}}\right) + \ec$ is upper bounded by $\ec$,
\begin{align}
\P\Bigg(\measuredangle(\mathbf{a},\mathbf{a}+\zb) \geq \cos^{-1}\bigg( \sqrt{\frac{\alpha}{\alpha + N}} \ \bigg) + \ec\Bigg) < \ec \ . 
\end{align}
\end{lemma} See Appendix \ref{a:angleproofs} for the proof.

\begin{figure}
\begin{center}
\includegraphics[width=2.7in]{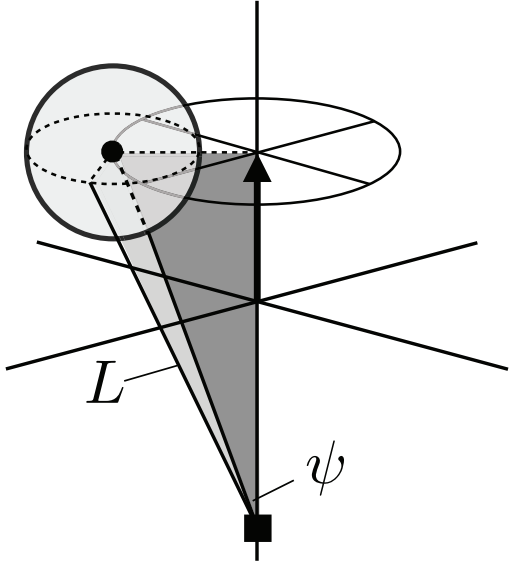}
\end{center}
\caption{With high probability, a Gaussian codeword, a Gaussian noise vector, and the red alert codeword vector are all nearly orthogonal to each other. Conditioning on this event, we can derive the minimum distance $L$ and the angle $\psi$ between a codeword plus noise and the red alert codeword.}\label{f:3DAngles4}
\end{figure}

In Figure \ref{f:3DAngles4}, we have depicted the distance $L$ and the angle $\psi$ from the red alert codeword to a standard codeword plus noise. Notice that both the noise and the codewords are (nearly) orthogonal to the axis along which the red alert codeword lies. 

Now consider a cone centered on the red alert codeword that contains a standard codeword plus noise with high probability. The next lemma upper bounds the required half-angle for the cone.

\begin{lemma}\label{l:angle}
Let $\mathcal{V}_n(\mathbf{x}(0),\mathbf{0},\psi)$ denote the cone centered on the red alert codeword with axis running towards the origin and half-angle $\psi$. For any $\ec > 0$, $w \in \{1,2,\ldots,e^{nR}\}$, and $n$ large enough, if the half-angle $\psi$ is greater than or equal to \begin{align}
 \sin^{-1}\left(\sqrt{\frac{\alpha \pa + N - \eb}{\pp + \pa + 2 \sqrt{\pp (1-\alpha)\pa} + N - \eb}}\right) + \ec \nonumber
\end{align} then the cone contains the codeword for message $w$ plus noise with high probability, i.e.,
\begin{align}
\P\Big(\mathbf{x}(w) + \mathbf{z} \in \mathcal{V}_n\big(\mathbf{x}(0),\mathbf{0},\psi \big)\Big) > 1 - \ec \ . 
\end{align}
\end{lemma} See Appendix \ref{a:angleproofs} for the proof.

\subsection{Red Alert Exponent}

Now that we know the decoding region can be confined to a conical shell, we can bound the probability of missed detection for the red alert codeword. 

\begin{lemma} \label{l:achievable} For any rate $R$, the following red alert exponent is achievable
\begin{align}
E(R)= \frac{\pp + \pa + 2\sqrt{\pp(\pa + N(1 - e^{2R}))}}{2N} - R \nonumber \ . 
\end{align}
\end{lemma}
\begin{proof}
Choose $\ec > 0$. In Lemma \ref{l:distance}, $L$ is a lower bound on the distance between the red alert codeword and a standard codeword plus noise. From Lemma \ref{l:angle}, we have an upper bound on the half-angle needed to capture a standard codeword plus noise in the cone $\Vm_n(\mathbf{x}(0),\mathbf{0},\psi)$ centered on the red alert codeword. If the received vector lies in the cone and is at least distance $L$ from the red alert codeword, then the decoder assumes the red alert message was not transmitted. Otherwise, it declares that the red alert message was sent. For $n$ large enough, we know that the probability that a random codeword plus noise, $\mathbf{x}(w) + \mathbf{z}$, leaves this region is at most $\epsilon$. Therefore, the probability of false alarm (averaged over the randomness in the codebook) is upper bounded by $\epsilon$.

If the received vector falls in the decoding region for standard messages, we simply subtract the offset $\sqrt{ (1-\alpha)\pa} \mathbf{1}$ and apply a maximum likelihood decoder to make an estimate of the transmitted message. Since the rate of the codebook is chosen to be slightly less than the capacity (for the power level $\alpha \pa - \eb$), it is straightforward to show that the average probability of error for a given message is at most $\epsilon$. 

Since the average false alarm probability and average error probability are small, it follows that there exists at least one fixed codebook with a small false alarm probability and average error probability. We now turn to upper bounding the probability of missed detection. Assume the red alert codeword is transmitted. Define
\begin{align}
\ee &= \frac{\pp + \pa + 2\sqrt{\pp(1-\alpha)\pa} - \eb - \ec}{N} \ . 
\end{align} where $\eb$ is specified by step 3) of the codebook construction in Section \ref{s:codebook}. Using Lemma \ref{l:gaussiannorm}, the probability that the noise pushes the red alert codeword further than $L$ (as specified in Lemma \ref{l:distance}) can be upper bounded by
\begin{align}
\P\big(\| \mathbf{x}(0) + \mathbf{z} \| > L\big) &= \P\big(\| \mathbf{x}(0) + \mathbf{z} \|^2 > nN(1 + \ee)\big)  \\ 
&\leq \exp\left(-\frac{n}{2}\big(\ee - \log(1 + \ee)\big)\right) \ . 
\end{align} The probability that the received vector falls into the cone of half-angle $\psi$ is given by the fraction of surface area of a sphere carved out by the cone. Using Lemma \ref{l:surfaceratio}, this can be calculated as 
\begin{align}
&\P\Big(\mathbf{x}(0) + \mathbf{z} \in \Vm_n\big(\mathbf{x}(0),\mathbf{0},\psi\big)\Big) \nonumber \\ & \qquad =  \frac{\sin^n\psi}{\sqrt{2\pi n} \sin \psi \cos \psi} \bigg(1 + O\left(\frac{1}{n}\right)\bigg)  \ . 
\end{align} Pulling terms into the exponent we get
\begin{align}
&\exp\bigg( -n \Big(-\log\big(\sin \psi\big) + \frac{1}{n} \log\big( \sqrt{2\pi n} \sin \psi \cos \psi\big) \nonumber\\
&\qquad\qquad~~~ + O\left(1/n\right)  \Big) \bigg) \ . 
\end{align} For $n$ large enough, we get that the probability is upper bounded by $\exp\Big( -n \big( -\log(\sin\psi) - \ec \big)  \Big)$.  

Since the noise is an i.i.d. Gaussian vector, its magnitude and direction are independent. Therefore, the probability of missed detection is upper bounded as
\begin{align}
p_{MD} &=\P\big(\| \mathbf{x}(0) + \mathbf{z} \| > L\big) \ \P\Big(\mathbf{x}(0) + \mathbf{z} \in \Vm_n\big(\mathbf{x}(0),\mathbf{0},\psi\big)\Big) \nonumber \\ &\leq  \exp\left(-\frac{n}{2}\Big(\ee - \log(1 + \ee) -2\log(\sin\psi) - 2\ec \Big)  \right) \nonumber
\end{align} for $n$ large enough. For $\eb$ and $\ec$ small enough and $n$ large enough, the exponent $\frac{\ee}{2} - \frac{1}{2}\log(1 + \ee) - \log(\sin\psi) - \ec$ can be made equal to
\begin{align}
&\frac{\pp + \pa + 2\sqrt{\pp(1-\alpha)\pa}}{2N} \nonumber \\ &~~-~ \frac{1}{2}\log\left(\frac{\pp + \pa + 2\sqrt{\pp(1-\alpha)\pa + N}}{N} \right)~  \nonumber\\ & ~~+~  \frac{1}{2}\log\left(\frac{\pp + \pa + 2\sqrt{\pp(1-\alpha)\pa + N}}{\alpha \pa + N} \right) - \ea \nonumber \\&=\frac{\pp + \pa + 2\sqrt{\pp(1-\alpha)\pa}}{2N} \nonumber \\ & ~~-~ \frac{1}{2}\log\left(1 + \frac{\alpha \pa}{N} \right) - \ea \nonumber \\&=\frac{\pp + \pa + 2\sqrt{\pp(1-\alpha)\pa}}{2N} - R - \ea \nonumber \ . 
\end{align} Finally, we can solve for $\alpha$ in terms of $R$ to get $\alpha = \frac{N}{\pa}(e^{2R} - 1)$. Substituting this into the expression above yields the desired result.
\end{proof}

Note that at $R = 0$, the coherent gain $2\sqrt{\pp(\pa + N(1 - e^{2R}))} = 2\sqrt{\pp \pa}$, which is the largest benefit we could hope for. At $R = C$, the coherent gain vanishes.

\begin{remark} We can interpret our achievability result from a hypothesis testing perspective. Let $\mathcal{H}_0$ denote the event that a standard codeword is transmitted and let $\mathcal{H}_1$ denote the event that the red alert codeword is transmitted. Under $\mathcal{H}_0$, the entries of $\mathbf{y}$ are i.i.d. according to a Gaussian distribution with mean $\sqrt{(1-\alpha)\pa}$ and variance $\alpha \pa + N$. Under $\mathcal{H}_1$, the entries are i.i.d. Gaussian with mean $-\sqrt{\pp}$ and variance $N$. Using the Chernoff-Stein Lemma \cite[Theorem 11.8.3]{coverthomas}, we can bound the missed detection probability of the optimal hypothesis test via the KL divergence between the two distributions, $D\big(\mathcal{N}(\sqrt{(1-\alpha)\pa}, \alpha \pa +N) \big\| \mathcal{N}(-\sqrt{\pp}, N)\big)$. A bit of calculation will reveal that this KL divergence corresponds exactly to the red alert exponent. One can obtain the same exponent by plugging these distributions into the red alert exponent expression from \cite[Lemma 1]{bnz09}. However, this does not in itself constitute a proof as the results of \cite{bnz09} are for DMCs without cost constraints. \end{remark}

\section{Converse} \label{s:converse}
We now develop an upper bound on the red alert exponent. Our bound relies on the fact that, in order to recover the standard messages reliably, we must allocate a significant volume of the output space for decoding them, which contributes to the probability of missed detection. An overview of the main steps in the proof is provided below.
\begin{itemize}
\item In Lemma \ref{l:shell}, we argue that a constant fraction of the codewords live in a thin shell and strictly satisfy the power and error constraints.
\item With high probability, the standard codewords plus noise are concentrated in a thin shell. Lemma \ref{l:volume} establishes this fact as well as the minimum volume required for the decoding region to attain a given probability of error.
\item To minimize the probability of missed detection, we should pack this volume into the thin shell to maximize the distance from the red alert codeword (see Figure \ref{f:converseregion} for an illustration). Lemma \ref{l:redalert} bounds the distance and angle from the red alert codeword to the resulting decoding region (see Figure \ref{f:converse} for an illustration).
\item Finally, in Lemma \ref{l:converse}, we bound the probability that the noise carries the red alert codeword into the decoding region for the standard codewords. 
\end{itemize}

\begin{lemma}\label{l:shell}
Assume that a sequence of codebooks satisfies the average block power constraint $\pa$ and has average probability of error $p_{\msg}$ that tends to zero. Then for any $\ef >0$ and $n$ large enough, there exists a shell of width $\ef$ that contains $e^{n(R-\ef)}$ codewords, each with probability of error at most $(2/\ef) p_{\msg} $, and average power at most $\pa(1 - \ef)^{-1}$.
\end{lemma} See Appendix \ref{a:converseproofs} for the proof.

\begin{lemma}\label{l:volume}
Assume that, for some $\gamma, \rho > 0$, $e^{n(R-\gamma)}$ codewords, each with probability of error at most $(2/\ef)p_{\msg}$ lie in the shell $\mathcal{T}_n(\mathbf{0},\sqrt{n \rho},\sqrt{n \rho} + \gamma)$. Then, for $n$ large enough, the decoding region for these codewords must include a subset of the noise-inflated shell $\mathcal{T}_n\big(\mathbf{0},\sqrt{n(\rho +N -\gamma)},\sqrt{n(\rho +N +\gamma)}\big)$ with volume at least
\begin{align*}
V_{\text{MIN}} = e^{n(R-\gamma)} \frac{(n\pi (N - \gamma))^{n/2}}{\Gamma\left(\frac{n}{2} + 1\right)}  \ . 
\end{align*} 
\end{lemma} See Appendix \ref{a:converseproofs} for the proof.

\begin{lemma}\label{l:redalert}
Assume that a sequence of codebooks has rate $R$ and an average probability of error $p_{\msg}$ that tends to zero as $n$ increases. Then, for sufficiently small $\ea$ and $n$ large enough, the probability of missed detection $p_{\md}$ is lower bounded by the probability that the noise vector has squared norm $\| \mathbf{z} \|^2$ between $L^2 + n\ea$ and $L^2 + 2n\ea$ and lies at an angle $\measuredangle(\mathbf{z},\mathbf{1})$ between $\psi(1 - \ea)$ and $\psi (1 -2 \ea)$ where
\begin{align}
L^2 ={n\bigg(\pp + \pa + N +  2 \sqrt{\pp\big(\pa + N(1 - e^{2R})\big)}  \bigg)}\nonumber \\
\psi = \sqrt{\frac{Ne^{2R}}{\pp + \pa + N + 2\sqrt{\pp(\pa + N(1-e^{2R}))}}}\ . \nonumber
\end{align}
\end{lemma} 
\begin{proof}
Consider the standard codewords from a red alert codebook. From Lemma \ref{l:shell}, for any $\ef > 0$ and $n$ large enough, at least $e^{n(R-\ef)}$ codewords with power at most $\pa(1 - \ef)^{-1}$ and probability of error at most $(2/\ef)p_{\msg}$ must lie in a shell $\mathcal{T}_n(\mathbf{0},\sqrt{n\rho}, \sqrt{n\rho} + \gamma)$ for some $\rho > 0$. From  Lemma \ref{l:volume}, it follows that the decoding region for these codewords falls within the noise-inflated shell $\mathcal{T}_n\big(\mathbf{0},\sqrt{n(\rho +N -\gamma)},\sqrt{n(\rho +N +\gamma)}\big)$ and has volume at least $V_{\text{MIN}}$.

\begin{figure}[h]
\begin{center}
\psset{unit=0.8mm}
\begin{pspicture}(-52,-44)(53,44)

\psframe[fillstyle=vlines,hatchcolor=gray,linestyle=none](25,-44)(55,44)
\pswedge[linewidth=1pt,linestyle=none,linecolor=white,fillcolor=white,fillstyle=solid](-50,0){87}{-32}{32}

\psframe[fillcolor=black,fillstyle=solid](-48,-2)(-52,2)

\psline[linewidth=1.5pt]{<-}(-50,2)(-40,15)
\rput(-40,23){Red Alert}
\rput(-40,18){Codeword}

\psarc[linestyle=solid,linewidth=2pt](15,0){38}{-67.9}{67.9}

\pscircle[linewidth=1pt](15,0){38}

\psarc[linestyle=solid,linewidth=2pt](-50,0){87}{-23.7}{-15.5}
\psarc[linestyle=solid,linewidth=2pt](-50,0){87}{15.5}{23.7}
\psarc[linestyle=solid,linewidth=2pt](15,0){30}{-52.1}{52.1}
\psarc[linestyle=solid,linewidth=1pt](15,0){30}{52.1}{-52.1}

\psarc[linestyle=dashed,linewidth=1.5pt](-50,0){87}{-30}{30}
\psline[linestyle=dashed,linewidth=1.5pt]{->}(-50,0)(37,0)
\rput(0,3){$d$}

\rput(49,0){\Large{$\mathcal{R}$}}
\rput(40,-40){\Large{$\mathcal{G}_d$}}

\psline[linewidth=1.5pt]{->}(-28,-28)(-15,-15)
\rput(-28,-32){Noise-Inflated Shell}

\end{pspicture}
\end{center}
\caption{To attain the desired probability of error, the decoding region for the standard codewords must include a subset of the noise-inflated shell with volume at least $V_{\text{MIN}}$. To minimize the probability of missed detection, we place this volume as far from the red alert codeword (square) as possible. Let $\mathcal{G}_d$ denote the set of points at distance $d$ or greater from the red alert codeword. The decoding region $\mathcal{R}$ is the intersection of $\mathcal{G}_d$ and the shell, where $d$ is chosen to capture volume $V_{\text{MIN}}$.}\label{f:converseregion}
\end{figure}
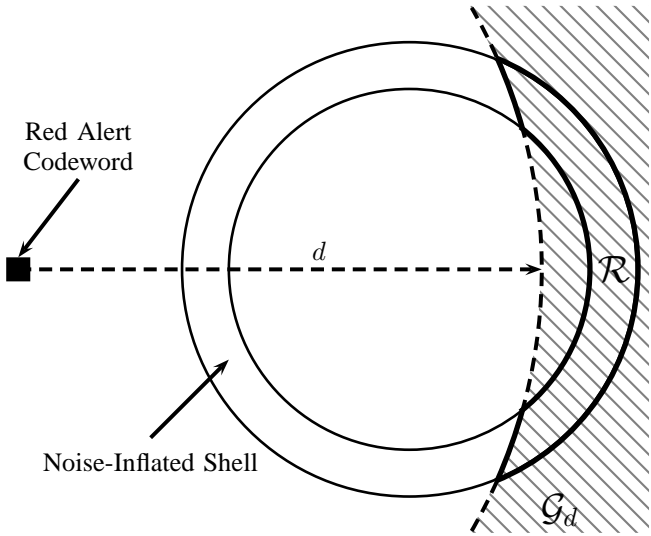

To get our lower bound, we need to pack this volume in the noise-inflated shell such that it minimizes $p_{\md}$. Since the noise vector is i.i.d. Gaussian, the probability that the red alert codeword is pushed to a certain point is determined solely by a decreasing function of the distance. Let $\mathcal{G}_d$ denote the set of all points at distance $d$ or greater from the red alert codeword
\begin{align}
\mathcal{G}_d = \{ \mathbf{z} : \| \mathbf{z} - \mathbf{x}(0) \| \geq d \} \ . 
\end{align}  The optimal volume packing corresponds to the intersection of the set $\mathcal{G}_d$ and the noise shell\\ $\mathcal{T}_n\big(\mathbf{0}, \sqrt{n(\rho +N- \ed)} , \sqrt{n(\rho +N+ \ed)}\big)$ with $d$ chosen such that the volume of the set is equal to $V_{\text{MIN}}$. Let $\mathcal{R}$ denote the resulting region and see Figure \ref{f:converseregion} for an illustration.

Let $\mathcal{R}_{\text{EDGE}}$ denote the set of points in $\mathcal{R}$ that sit at the minimum distance to the red alert codeword,
\begin{align}
\mathcal{R}_{\text{EDGE}} = \Big\{ \mathbf{u} \in \mathcal{R} : \| \mathbf{u} -  \mathbf{x}(0)\| = \min_{\mathbf{w} \in \mathcal{R}} \| \mathbf{w} - \mathbf{x}(0)\| \Big\} \ , \nonumber
\end{align} and let $\mathbf{v}^* \in \mathcal{R}_{\text{EDGE}}$ be any of these points. Let $L^*$ and $\psi^*$ denote the distance and angle from the red alert codeword $\mathbf{x}(0)$ to $\mathbf{v}^*$. We now seek to bound these quantities through a bound on the angle from the origin to $\mathbf{v}^*$.

\begin{figure}[h]
\begin{center}
\psset{unit=0.8mm}
\begin{pspicture}(-52,-38)(53,38)

\psframe[fillcolor=black,fillstyle=solid](-48,-2)(-52,2)

\psarc[linestyle=solid,linewidth=2pt](15,0){38}{-67.9}{67.9}

\pscircle[linewidth=1pt](15,0){38}
\pscircle[linewidth=1pt,linestyle=none,fillstyle=solid,fillcolor=white](15,0){30}

\pswedge[linewidth=1pt,linestyle=solid,linecolor=black,fillcolor=gray!15!white,fillstyle=solid](15,0){38}{-59}{59}
\psarc[linestyle=solid,linewidth=2pt](-50,0){87}{-23.7}{-15.5}
\psarc[linestyle=solid,linewidth=2pt](-50,0){87}{15.5}{23.7}
\psarc[linestyle=solid,linewidth=2pt](15,0){30}{-52.1}{52.1}
\psarc[linestyle=solid,linewidth=1pt](15,0){30}{52.1}{-52.1}

\psarc[linestyle=dashed,linewidth=1pt](-50,0){92}{-20}{20}
\psarc[linestyle=dashed,linewidth=1pt](-50,0){89}{-20}{20}
\psline[linestyle=dashed,linewidth=1pt](-50,0)(34.5,32.8)
\psline[linestyle=dashed,linewidth=1pt](-50,0)(34.5,-32.8)

\pscircle[linewidth=1.5pt,fillstyle=solid,fillcolor=white](15,0){1.3}

\psline(-50,0)(45,0)

\psarc(-50,0){19}{0}{20.5}
\rput(-28,4){{$\psi$}}
\psarc(15,0){8}{0}{59}
\rput(25,6){{$\theta$}}

\rput(0,22.5){{$L$}}
\rput{59}(20,14){\scriptsize{$\sqrt{n(\rho + N + \ed)}$}}

\rput(29.5,35){\pscircle[fillstyle=solid,fillcolor=black](0,0){1.2}}
\rput(30.5,39){$\mathbf{v}^*$}
\rput(34.5,32.8){\pscircle[fillstyle=solid,fillcolor=black](0,0){1.2}}
\rput(34.5,36.3){$\mathbf{v}$}

\rput{16.5}(35.7,23.9){\psframe[linestyle=none,fillstyle=solid,fillcolor=black!50!white](0,0)(3,7)}
\rput{-16.5}(35.7,-23.9){\psframe[linestyle=none,fillstyle=solid,fillcolor=black!50!white](0,-7)(3,0)}
\rput(49,0){\Large{$\mathcal{R}$}}

\end{pspicture}
\end{center}
\caption{Illustration of successive lower bounds on the probability of missed detection, $p_{\md}$. The red alert codeword is represented by a square and the origin by a circle. The decoding region $\mathcal{R}$ is denoted by a thick line. We would like to characterize the distance $L^{*}$ and angle $\psi^{*}$ to the edge of $\mathcal{R}$, represented by the point $\mathbf{v}^{*}$. To do so, we consider a cone with half-angle $\theta$ (shaded region) with the same volume as $\mathcal{R}$. The intersection of this cone with the outer surface of $\mathcal{R}$ contains a point $\mathbf{v}$ at distance $L > L^{*}$ and angle $\psi < \psi^*$. In our final lower bound, we only consider the event that the received vector lies in the subset of $\mathcal{R}$ at distance slightly larger than $L$ and within an angle between $\psi - \epsilon$ to $\psi - 2 \epsilon$ from the red alert codeword (illustrated by dark patches).
}\label{f:converse}
\end{figure}
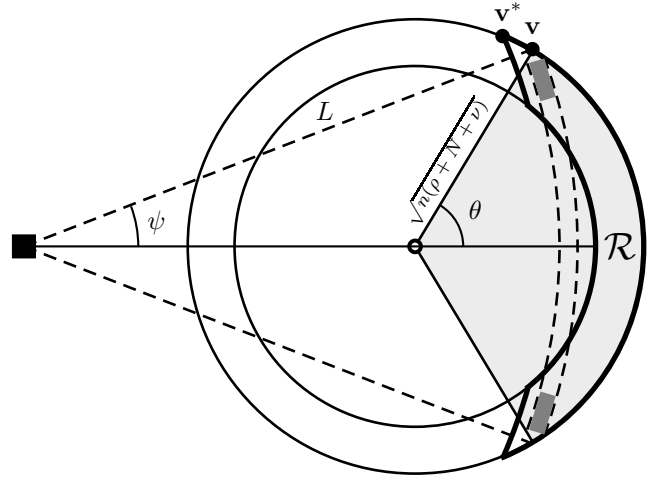

Let $\theta^*$ denote the half-angle of a cone, centered on the origin that contains the region $\mathcal{R}$ (and thus includes $\mathbf{v}^*$). The volume of this cone must be at least equal to that of $\mathcal{R}$ since $\mathcal{R}$ is a subset of the noise shell. Therefore, $\theta^*$ is lower bounded by the half-angle $\theta$ of a cone whose volume is equal to the volume of $\mathcal{R}$ (see Figure \ref{f:converse} for an illustration). Combining (\ref{e:spherevol}) and Lemma \ref{l:surfaceratio}, for $n$ large enough, the volume of this cone is upper bounded by \begin{align}
\frac{(n\pi(\rho+N + \ed) \sin^2 \theta)^{n/2}}{\Gamma\left(\frac{n}{2} + 1\right)} \ .
\end{align} Now, since we require this quantity to exceed $V_{\text{MIN}}$, we can lower bound $\theta$ by
\begin{align}
\theta \geq \sin^{-1}\left(e^{R-\ed} \sqrt{\frac{N - \ed}{\rho + N + \ed}} \right) \ .
\end{align} We can further lower bound $\theta$ by setting $\rho$ to its maximum value $\pa (1- \gamma)^{-1}$,
\begin{align}
\theta \geq  \sin^{-1}\left(e^{R-\ed} \sqrt{\frac{N - \ed}{\pa(1 - \gamma)^{-1}+N + \ed}}\right)  \ . 
\end{align} Thus, for any $\ec > 0$, $\ed$ and $\gamma$ small enough, and $n$ large enough, $\theta^*$ is lower bounded as follows
\begin{align}
\theta^* \geq \theta \geq   \sin^{-1}\left(e^{R-\ec} \sqrt{\frac{N}{\pa + N}} \right) \label{e:theta} \ .
\end{align}

The distance $L^*$ from $\mathbf{x}(0)$ to $\mathbf{v}^*$ is upper bounded by the distance $L$ to a point $\mathbf{v}$ that lies on the intersection of the outer shell (at distance $\sqrt{n(\pa(1 - \gamma)^{-1} + N + \ed)}$ from the origin) and the cone of half-angle $\theta$. Without loss of generality, assume that the red alert codeword is placed at $\mathbf{x}(0) = -\mu \mathbf{1}$ for some $\mu > 0$.\footnote{It is straightforward to prove that placing the red alert codeword exactly at the origin is suboptimal.} The direction of the red alert codeword is not important since we will always fill the noise shell relative to this direction. Then $L^2$ is at least 
\begin{align}
&\left(\sqrt{n\mu} + \cos\theta \sqrt{n(\pa(1-\gamma)^{-1} + N +\ed)}\right)^2 \nonumber\\ &~~~~+~\left(\sin \theta \sqrt{n(\pa(1-\gamma)^{-1}+ N +\ed)}\right)^2  \ .
\end{align} For any $\ec > 0$, $\gamma$ and $\ed$ small enough, and $n$ large enough, this quantity is itself upper bounded by
\begin{align}
n\left(\mu + \pa + N + 2 \cos \theta \sqrt{\mu(\pa + N)} + \ec\right)  \nonumber \ .
\end{align} 

The half-angle $\psi$ of a cone, centered on the red alert codeword, that contains the point $\mathbf{v}$ is lower bounded by 
\begin{align}
\sin^{-1} \left(\frac{\sin \theta \sqrt{n(\pa(1-\gamma)^{-1} + N + \nu)}}{\sqrt{n\left(\mu + \pa + N + 2 \cos \theta \sqrt{\mu(\pa + N)} + \ec\right)}}\right) \nonumber
\end{align} which, for any $\eta > 0$, $\gamma$, $\ed$, and $\ec$ small enough, and $n$ large enough is itself lower bounded by
\begin{align}
\sin^{-1} \left(\frac{(1 - \eta)\sin \theta \sqrt{\pa + N}}{\sqrt{\mu + \pa + N + 2 \cos \theta \sqrt{\mu(\pa + N)}}}\right) \nonumber \ . 
\end{align} 

The probability of missed detection decreases if the distance $L^*$ from $\mathbf{x}(0)$ to $\mathbf{v}^*$ is increased. The angle $\psi^*$ will simultaneously decrease. Thus, by setting $\mu = \pp$, we further lower bound the probability of missed detection. Using the relation $\sin^2 \theta + \cos^2 \theta = 1$ combined with (\ref{e:theta}), we find that $\cos \theta \leq \sqrt{1 - e^{2R-2\ec} \frac{N}{N + \pa}}$. Plugging in $\mu$ and $\theta$, we obtain the following upper bound on $(L^*)^2$:
\begin{align}
n\Big(\pp + \pa + N +  2 \sqrt{\pp(\pa + N(1 - e^{2R-2\ec})}) + \ec\Big)  \nonumber 
\end{align} and the following lower bound on $\psi^*$:
\begin{align}
\sin^{-1} \sqrt{\frac{(1 - \eta)^2N e^{2R - 2\ec}}{\pp + \pa + N +  2 \sqrt{\pp(\pa + N(1 - e^{2R - 2 \ec})}) }} \nonumber 
\end{align}

Finally, it follows that, for $\ea$ small enough (but greater than $0$ for finite $n$) and $n$ large enough, the optimal packing contains all points from squared distance $L^2 + n \ea$ to $L^2 + 2n \ea$ from the red alert codeword and angle $\psi (1 - \ea)$ to $\psi (1 - 2 \ea)$ where $L$ and $\psi$ are as in the statement of the theorem. Thus, the probability of missed detection is lower bounded by the event that the noise falls into this region.
\end{proof}

\begin{lemma}\label{l:converse} For any rate $R$, the red alert exponent is upper bounded by
\begin{align}
E(R) \leq \frac{\pp + \pa + 2\sqrt{\pp(\pa + N(1 - e^{2R}))}}{2N} - R \nonumber \ . 
\end{align}
\end{lemma}
\begin{proof}
Lemma \ref{l:redalert} established that the probability of missed detection is lower bounded by the event that the noise has squared length between $L^2 + n\ea$ to $L^2 + 2n\ea$ and angle between $\psi(1 - \ea)$ and $\psi (1 - 2\ea)$ for some $\ea$ that tends to $0$ as $n$ tends to infinity. We now lower bound the probability of this event. Define
\begin{align}
\beta =  \frac{\pp + \pa + 2\sqrt{\pp(\pa + N(1 - e^{2R}))}}{N} \ . 
\end{align}

Since the magnitude and angle of an i.i.d. Gaussian vector are independent, the probability of missed detection is lower bounded as follows:
\begin{align}
p_{\md} &\geq \P\left(L^2 + n\ea \leq \|\mathbf{z}\|^2 \leq L^2 + 2n\ea \right) \nonumber \\ & ~~~~ \cdot~ \P\Big(\mathbf{z} \in \big\{\mathcal{V}_n(\mathbf{0},\mathbf{1},\psi(1 - \ea))\setminus \mathcal{V}_n(\mathbf{0},\mathbf{1},\psi(1 - 2\ea)) \big\} \Big) \nonumber
\end{align} By Lemma \ref{l:surfaceratio}, for $n$ large enough, the second term in the product can be lower bounded by 
\begin{align}
\big(\sin(\psi (1 - \ea)^2)\big)^n - \big(\sin(\psi (1 -2 \ea))\big)^n  
\end{align} which, for $n$ large enough, is itself lower bounded by 
\begin{align}
&\big((1-\ea)^3 \sin(\psi (1 - \ea))\big)^n \\
&=\exp\Big( -n \big(-\log\big((1-\ea)^3 \sin (\psi(1 - \ea))\big)  \Big) \ . 
\end{align}
Now, substituting in the lower bound on $\psi$ from Lemma \ref{l:redalert}, we arrive at the following lower bound
\begin{align}
\exp\Big(-n \Big( -R + \frac{1}{2}\log(1 + \beta) - 3 \log(1 -\ea) \Big)\Big) \label{e:angleexp} \ . 
\end{align}

Using the upper bound on $L$ from Lemma \ref{l:redalert} and applying Theorem \ref{t:cramer} for Chi-square random variables (and noting that $\ea$ and $\ed$ go to zero as $n$ goes to infinity), it follows that 
\begin{align}
&\liminf_{n\rightarrow\infty} \bigg( -\frac{1}{n}\log\P\left(L^2 + n\ea \leq \|\mathbf{z}\|^2 \leq L^2 + 2n\ea \right) \bigg) \\
&\geq \frac{\beta}{2} - \frac{1}{2}\log(1 + \beta)\ . 
\end{align} Combining this with the lower bound on the angle event in (\ref{e:angleexp}), the exponent of the probability of missed detection is lower bounded by
\begin{align}
&\liminf_{n\rightarrow\infty} \bigg( -\frac{1}{n}\log p_{\md} \bigg) \\&~\geq  \frac{\beta}{2} - \frac{1}{2}\log(1 + \beta) + \frac{1}{2} \log(1 + \beta) - R \\
&~= \frac{\pp + \pa + 2\sqrt{\pp(\pa + N(1 - e^{2R}))}}{2N} - R  \nonumber 
\end{align}as desired.
\end{proof} 

\section{Plots}

\begin{figure}[h!]
\centering
\includegraphics[width=3.8in]{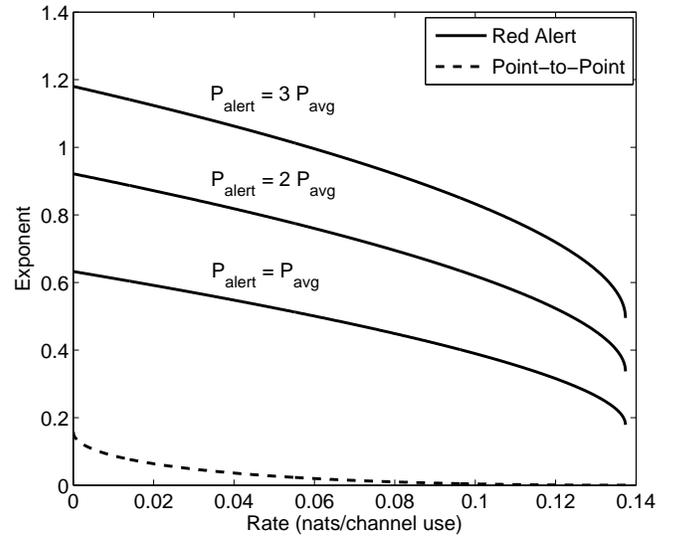}
\caption{Optimal red alert exponent for $\pa = -5$dB with $\pp = \pa, 2\pa, 3\pa$. An upper bound on the point-to-point AWGN error exponent is provided for comparison. }\label{f:redalertplot1}
\end{figure}

\begin{figure}[h!]
\centering
\includegraphics[width=3.8in]{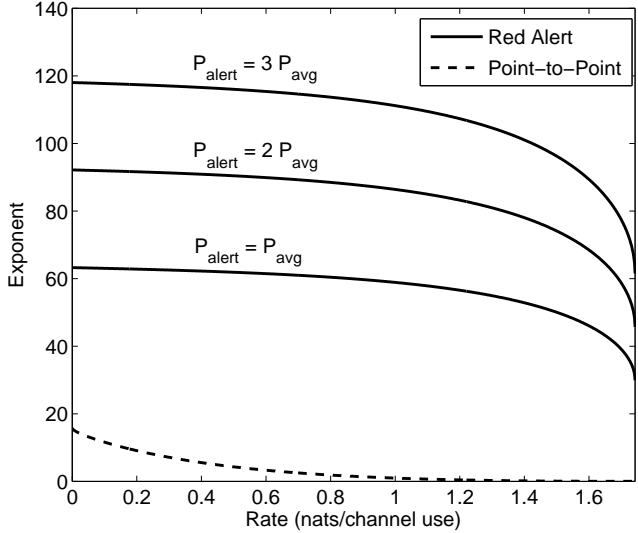}
\caption{Optimal red alert exponent for $\pa = 15$dB with $\pp = \pa, 2\pa, 3\pa$. An upper bound on the point-to-point AWGN error exponent is provided for comparison. }\label{f:redalertplot2}
\end{figure}
In Figures \ref{f:redalertplot1} and \ref{f:redalertplot2}, we have plotted the optimal red alert exponent for $\pa = -5$dB and $\pa = 15$dB, respectively, with red alert power constraints $\pp = \pa, 2\pa$, and $3 \pa$. For comparison, we have also plotted an upper bound on the AWGN point-to-point error exponent from \cite[Equation 4]{shannon59}. Notice that the red alert exponent can be quite large at capacity, even when the red alert power constraint is equal to the average power constraint.

\section{Conclusions}
We have developed sharp bounds on the error exponent for distinguishing a single special message from $2^{nR}$ standard messages over an AWGN channel. As discussed in the introduction, these bounds can be used to characterize the performance of certain data streaming architectures, where each bit must be decoded after a given delay. An interesting question for future study is how well a single special message can be protected at a given finite blocklength, i.e., understanding the limits of unequal error protection in the non-asymptotic regime \cite{ppv10}.  
\appendices
\section{Proofs for Section \ref{s:distance}} \label{a:distanceproofs}

\begin{proof}[Proof of Lemma \ref{l:gaussiannorm}]
The squared Euclidean distance is the sum of $n$ i.i.d. squared Gaussian random variables with variance $N$. Therefore, $\frac{1}{nN}\| \mathbf{z} \|^2$ is the sum of $n$ i.i.d. Chi-square random variables. Applying Theorem \ref{t:cramer} and plugging in the Chi-square rate function of $I_Z(b) = \frac{1}{2}(b - 1 - \log b)$, it follows that 
\begin{align}
\P\big(\| \mathbf{z} \|^2 \geq nNb\big) \leq 2 \exp\left(-\frac{n}{2}(b - 1- \log b)\right) \ .
\end{align} Substituting in $b = 1 + \beta$ yields the first bound and $b = 1 - \beta$ yields the second.
\end{proof}

\begin{proof}[Proof of Lemma \ref{l:gaussianortho}]
First, we write the probability that $| \mathbf{a}^T\mathbf{z} | $ is greater than $t$ in terms of the Q-function,
\begin{align}
\P\big( | \mathbf{a}^T \mathbf{z} | \geq t \big) &= 2 Q\left(\frac{t}{\sqrt{N}\| \mathbf{a} \|}\right) \\&<   \exp\left(-\frac{t^2}{2N\| \mathbf{a} \|^2}\right) \ . 
\end{align} Substituting $t = \ec \| \mathbf{a}\|^2$ yields, 
\begin{align}
\P\big( | \mathbf{a}^T \mathbf{z} | \geq  \ec \| \mathbf{a}\|^2 \big) < \exp\left(-\frac{\ec^2 n\alpha}{2N} \right),
\end{align} which can be driven arbitrarily close to zero for $n$ large enough.
\end{proof}

\begin{proof}[Proof of Lemma \ref{l:distance}]
We simply wish to bound the length of the vector from the special codeword to a standard codeword plus noise, $-\mathbf{x}(0) + \mathbf{x}(w) + \mathbf{z}$. By expanding terms, we obtain:
\begin{align}
&\| -\mathbf{x}(0) + \mathbf{x}(w) + \mathbf{z} \|^2 \\& = \left\| \left(\sqrt{\pp} + \sqrt{(1-\alpha)\pa}\right) \mathbf{1}+ \mathbf{v}(w) + \mathbf{z}\right\|^2 \\ &= \left(\sqrt{\pp} +  \sqrt{(1-\alpha)\pa}\right)^2 \mathbf{1}^T \mathbf{1}  \nonumber \\
& ~~~+ ~2\left(\sqrt{\pp}+ \sqrt{(1-\alpha)\pa}\right) \mathbf{1}^T (\mathbf{v}(w) + \mathbf{z})   \nonumber \\
& ~~~+~  \| \mathbf{v}(w) + \mathbf{z}\|^2 \ . 
\end{align} The first term is $n( \pp + (1-\alpha)\pa + 2 \sqrt{\pp (1-\alpha)\pa}$. The second term is the inner product of a fixed vector $2\left(\sqrt{\pp}+ \sqrt{(1-\alpha)\pa}\right) \mathbf{1}$ and an i.i.d. Gaussian vector $\mathbf{v}(w) + \mathbf{z}$ since $\mathbf{v}(w)$ is an element of a random Gaussian codebook. Thus, using Lemma \ref{l:gaussianortho}, it can be shown that the probability that this inner product is less than $-n\ec/2$ is at most $\ec/2$ for $n$ large enough. The third term is the squared norm of an i.i.d. Gaussian vector with mean zero and variance $\alpha \pa  + N - \eb$. From Lemma \ref{l:gaussiannorm}, it follows that $\| \mathbf{v}(w) + \mathbf{z}\|^2$ is less than $n( \alpha \pa + N - \eb - \ec/2)$ with probability at most $\ec/2$ for $n$ large enough. Combining these three bounds completes the proof.
\end{proof}

\section{Proofs for Section \ref{s:angle}} \label{a:angleproofs}

\begin{proof}[Proof of Lemma \ref{l:gaussianangle}]
The angle between $\mathbf{a}$ and $\mathbf{a}+ \mathbf{z}$ is
\begin{align}
\measuredangle(\mathbf{a},\mathbf{a}+\mathbf{z}) = \cos^{-1} \left( \frac{\mathbf{a}^T (\mathbf{a}+\mathbf{z})}{\| \mathbf{a}\| \|\mathbf{a} + \mathbf{z}\|} \right) \ .
\end{align} From Lemma \ref{l:gaussianortho}, for any $\ed > 0$ and $n$ large enough, the probability that $\mathbf{a}^T \mathbf{z} > \ed \| \ab \|^2$ is at most $\ed$. Therefore, since $\| \mathbf{a}\| = n \alpha$ we have that $\ab^T (\ab + \zb) \geq (1 + \ed) n \alpha$ with probability at most $\ed$. Combining Lemmas \ref{l:gaussiannorm} and \ref{l:gaussianortho}, we can also show that the probability that $\| \ab + \zb \| < (1 - \ed) \sqrt{n(\alpha + N)}$ is at most $\ed$ for $n$ large enough. Thus, the probability that 
\begin{align}
\frac{\mathbf{a}^T (\mathbf{a}+\mathbf{z})}{\| \mathbf{a}\| \|\mathbf{a} + \mathbf{z}\|} &< \frac{(1+\ed)n\alpha }{\sqrt{n\alpha}(1- \ed)\sqrt{n(\alpha + N)}} \\&  = \frac{1+\ed}{1-\ed} \sqrt{\frac{\alpha}{\alpha + N}} 
\end{align} is at most $2 \ed$. Choosing $\ed$ small enough yields the desired result.
\end{proof}

\begin{proof}[Proof of Lemma \ref{l:angle}]
The angle between the axis of the cone and the standard codeword plus noise is
\begin{align}
\measuredangle(-\mathbf{x}(0),-\mathbf{x}(0) + \mathbf{x}(w) + \mathbf{z}) = \cos^{-1}\left( u \right) \\
u = \frac{(-\mathbf{x}(0))^T(-\mathbf{x}(0) + \mathbf{x}(w) + \mathbf{z})}{\| \mathbf{x}(0) \| \| -\mathbf{x}(0) + \mathbf{x}(w) + \mathbf{z}\|}  \ . 
\end{align} Since $\cos^{-1}(u)$ is a decreasing function of $u$, an upper bound on the angle can be obtained by lower bounding $u$. We will do this by lower bounding the numerator and upper bounding the denominator (with high probability). Expanding the numerator yields:
\begin{align}
&(-\mathbf{x}(0))^T(-\mathbf{x}(0) + \mathbf{x}(w) + \mathbf{z}) \\
& = \sqrt{\pp}\left(\sqrt{\pp} + \sqrt{(1-\alpha)\pa}\right) \mathbf{1}^T \mathbf{1} \\
&~~~~+~  \sqrt{\pp} \ \mathbf{1}^T ( \mathbf{v}(w)+ \mathbf{z}) \ . 
\end{align} The first term is simply $n\sqrt{\pp}\big(\sqrt{\pp} + \sqrt{(1-\alpha)\pa} \big)$. The second term is the inner product of a fixed vector and an i.i.d. Gaussian vector. Thus, using Lemma \ref{l:gaussianortho}, it can be shown that for any $\ed > 0$ and $n$ large enough, the probability that the second term is less than $-\ed n$ is at most $\ed$. The denominator is composed of two terms. The first is simply $\| \mathbf{x}(0) \| = \sqrt{n\pp}$. Following the proof of Lemma \ref{l:distance}, it can be shown that the second term $\|-\mathbf{x}(0) + \mathbf{x}(w) + \mathbf{z}\|$ is greater than 
\begin{align}
\sqrt{n \left( \pp + \pa + N + 2 \sqrt{\pp(1-\alpha)\pa} - \eb + \ed\right)} \nonumber
\end{align} with probability at most $\ed$. Combining these bounds, we get that the probability that $u$ is less than 
\begin{align}
\frac{\sqrt{\pp} + \sqrt{(1-\alpha)\pa} - \ed}{\sqrt{ \pp + \pa + N + 2 \sqrt{\pp(1-\alpha)\pa} - \eb + \ed} }
\end{align} with probability at most $2\ed$. Thus, so long as the half-angle $\psi$ is greater than or equal to 
\begin{align}
\cos^{-1}\Bigg(\frac{\sqrt{\pp} + \sqrt{(1-\alpha)\pa}}{\sqrt{ \pp + \pa + 2 \sqrt{\pp(1-\alpha)\pa}+ N  - \eb} }\Bigg)  + \ec \nonumber \end{align} the cone contains $\mathbf{x}(w) + \mathbf{z}$ with probability at least $1 - \ec$ for $n$ large enough. Applying the trigonometric identity $\sin^2(\psi) + \cos^2(\psi) = 1$ completes the proof.
\end{proof}

\section{Proofs for Section \ref{s:converse}} \label{a:converseproofs}

\begin{proof}[Proof of Lemma \ref{l:shell}]Observe that at least one codeword has power at most $\pa$, otherwise the average will be larger than $\pa$. If we remove this codeword's contribution from the average, the remaining codewords have average power at most $\pa \frac{e^{nR}}{e^{nR} - 1}$. Now, we can find a codeword whose power must be at most $\pa \frac{e^{nR}}{e^{nR} - 1}$. Removing this codeword yields an average of $\pa \frac{e^{nR}}{e^{nR} - 2}$. Continuing this process, we can remove $\ef e^{nR}$ codewords that each have power at most $\pa (1 - \ef)^{-1}$. By the same argument, we can find $(1- (\gamma/2))e^{nR}$ codewords that each have probability of error at most $(2/\ef)p_{\msg}$. Therefore, at least $(\gamma/2) e^{nR}$ codewords must satisfy both these constraints simultaneously.

The selected codewords live in the sphere of radius $\sqrt{n\pa(1-\gamma)^{-1}}$. We partition this sphere into shells of width $\gamma$ each. It follows that at least one of these shells must contain $\frac{\gamma^2}{2\sqrt{n(\pa(1 - \gamma)^{-1})}} e^{nR}$ codewords. Finally, select $n$ large enough so that $e^{n(R-\gamma)} \leq \frac{\gamma^2}{2\sqrt{n(\pa(1 - \gamma)^{-1})}} e^{nR} $.
\end{proof}

\begin{proof} [Proof of Lemma \ref{l:volume}]
Assume that one of the codewords from the shell $\mathcal{T}_n(\mathbf{0},\sqrt{n \rho},\sqrt{n \rho} + \gamma)$ is transmitted. It follows from Lemma \ref{l:gaussiannorm}, that for any $\gamma > 0$ and $n$ large enough, the probability that $\| \mathbf{y} \|$ is larger than $\sqrt{n(\rho +N+ \gamma)}$ or smaller than $\sqrt{n(\rho+N-\gamma)}$ is upper bounded by $p_{\msg} \frac{1}{\gamma}$. If the noise lands outside this ``noise shell,'' then we will assume that the transmitted codeword is decoded correctly. However, each codeword still needs to capture $1 - p_{\msg} \frac{3}{\gamma}$ probability inside the shell to ensure the error probability does not exceed $p_{\msg} \frac{2}{\gamma}$. 

Now, consider the volume required for decoding a single codeword reliably. Since the noise is i.i.d. Gaussian, its probability distribution is rotationally invariant. This implies that the shape that uses the least volume to capture a given probability of error is a sphere centered on the codeword. Let $\sqrt{n \ed}$ be the radius of this sphere. By Lemma \ref{l:gaussiannorm}, if $\ed < N$, the probability that the noise falls inside this sphere goes to zero exponentially in $n$ which implies the probability of error goes to one. Therefore, for $n$ large enough, the probability of error will always exceed the desired probability of error (which is assumed to be bounded away from one). Using (\ref{e:spherevol}), we get that the decoding region of each codeword must have volume at least $\frac{(n\pi(N-\gamma))^{n/2}}{\Gamma\left(\frac{n}{2} + 1\right)}$ for any $\gamma > 0$. We find that we will need a volume of at least 
\begin{align}
V_{\text{MIN}} = e^{n(R-\gamma)} \frac{(n\pi (N - \gamma))^{n/2}}{\Gamma\left(\frac{n}{2} + 1\right)} \label{e:vmin}
\end{align} to reliably decode these codewords. 
\end{proof}

\section{Offset Codes Versus Conical Codes} \label{a:cone}

We now develop some intuition for why the offset construction
of Section~\ref{s:codebook} is a better construction than the conical construction we used
in our earlier work~\cite{nd10allerton}. The difference between these two constructions is easier to understand in a discrete setting so we will analyze the corresponding constructions for a BSC with crossover probability $p$. For ease of analysis, we will calculate rate in bits per channel use (rather than nats per channel use).

First, recall that the BSC red alert exponent can be attained using a fixed composition codebook. Specifically, each of the $2^{nR}$ codewords is drawn independently and uniformly from the set of weight-$nq$ binary sequences. If the rate is less than the induced mutual information, the average probability of error can be driven to zero
\begin{align}
R < I(X;Y) = h_B(q*p) - h_B(p) \ . 
\end{align} The red alert codeword is taken to be all the all zeros vector. The decoder runs a hypothesis test between the two possible output distributions, Bernoulli$(p)$ and Bernoulli$(q*p)$. The error exponent for the probability of missed detection is the KL divergence between the two distributions, $D(q*p \| p).$ As shown in \cite{sd08}, this is the optimal red alert exponent.

We can construct a \textit{conical code} of parameter $q > 1/2$ by first drawing $2^{n(C-\epsilon)}$ codewords i.i.d. according to a Bernoulli$(1/2)$ distribution for some $\epsilon > 0$. Let $\mathcal{C}$ denote the resulting set of codewords. To guarantee the same red alert exponent, we only keep those codewords with Hamming weight $nq$ or greater and set the red alert codeword to be the all zeros vector. We now bound the rate of this construction. Using Theorem \ref{t:cramer}, it can be shown that the probability that a Bernoulli$(1/2)$ sequence $\mathbf{x}$ has Hamming weight at least $nq$ is upper bounded as 
\begin{align}
\P({\rm{wt}}(\mathbf{x}) \geq nq) \leq 2^{-n D(q\|0.5)} \label{e:weightprob} \ . 
\end{align}

Take $\mathcal{C}_q$ to be the set of subset of codewords in $\mathcal{C}$ with Hamming weight $nq$ or greater. Using (\ref{e:weightprob}), the expected size of $\mathcal{C}_q$ is upper bounded by
\begin{align}
\E\big[ |\mathcal{C}_q|\big] &\leq 2^{n(C-D(q\|0.5) -\epsilon)} \\
&= 2^{n(1 - h_B(p) - (1 - h_B(q)) -\epsilon)} \\
&= 2^{n(h_B(q) - h_B(p) - \epsilon)} \ .
\end{align} It can be shown with a Chernoff bound that the probability $\mathcal{C}_q$ contains significantly more codewords vanishes doubly exponentially in $n$. Furthermore, it can be shown that the average probability of error of $\mathcal{C}_q$ vanishes with $n$. Therefore, the rate of the conical codebook is $h_B(q) - h_B(p)$ for a red alert exponent of $D(q*p \| p)$. 

Now, observe that $q < q*p < 1/2$ (unless either $q$ or $p$ is equal to $1/2$) so $h_B(q) < h_B(q*p)$, meaning that the rate of the conical construction is strictly less than the (optimal) constant composition construction. Intuitively, this means that the usual i.i.d. Bernoulli$(\frac{1}{2})$ construction used to approach capacity does not pack codewords of higher (or lower) weights efficiently. Constraining the weight of codewords is essential to the hypothesis test that leads to the red alert exponent. The constant composition (or offset) construction is successful since it optimizes the packing of codewords of a given weight. A similar phenomenon occurs in the AWGN setting as shown below.

\section{AWGN Conical Codes: Red Alert Exponent} \label{a:coneexponent}

For completeness, we review the AWGN conical code that we proposed in \cite{nd10allerton} and the resulting red alert exponent. The construction is comprised of three main steps:
\begin{enumerate}
\item Place the red alert codeword at the limit of the red alert power constraint, $\mathbf{x}(0) = -\sqrt{\pp}\mathbf{1}$.
\item Draw $2^{n(C - \epsilon)}$ codewords i.i.d. according to a Gaussian distribution with mean $0$ and variance $\pa - \epsilon$.
\item Of these codewords, only keep the first $2^{nR}$ that lie in the cone $\mathcal{V}_n(\mathbf{0},\mathbf{1},\theta+\epsilon)$ where $\theta = \sin^{-1}\left(e^{-(C-R)}\right)$. (If there are fewer than $2^{nR}$ such codewords, declare an error.) 
\end{enumerate}

It can be shown that with high probability the resulting codebook contains $2^{nR}$ codewords inside the cone of half-angle $\theta$. We now turn to bounding the distance and angle from the standard decoding region to the red alert codeword. 

The distance can be bounded using the techniques used to prove Lemma \ref{l:distance}. It follows that for any $\ec > 0$ and $n$ large enough, the squared distance from the red alert codeword to a standard codeword plus noise is at least $L$ with high probability,
\begin{align}
&\P\big( \| -\mathbf{x}(0) + \mathbf{x}(w) + \mathbf{z} \| \geq L \big) > 1 - \ec \\
&L^2 = n\big(\pp + \pa + 2\sqrt{\pp \pa} \cos \theta + N - \ec \big) \ .
\end{align} Substituting in $\cos^2\theta = 1 - \sin^2 \theta = 1 - e^{-2(C-R)}$, we get that $L^2$ is equal to
\begin{align}
n\Big(\pp + \pa + N+  2\sqrt{\pp \pa\left(1 -e^{-2(C-R)}\right)}  - \ec \Big) \ . \nonumber
\end{align}

Similarly, the techniques from Lemma \ref{l:angle} can be used to bound the angle. Let $\mathcal{V}_n(\mathbf{x}(0),\mathbf{0},\psi)$ denote the cone centered on the red alert codeword with axis running towards the origin and half-angle $\psi$. For any $\ec > 0$ and $n$ large enough, if the half-angle $\psi$ is larger than
\begin{align}
 \sin^{-1}\left(\frac{e^{-(C-R)} \sqrt{\pa}}{\sqrt{\pp + \pa + N + 2\sqrt{\pp \pa\left(1 -e^{-2(C-R)}\right)}}}\right)\nonumber
\end{align} then the cone contains the codeword for message $w$ plus noise with high probability, i.e.,
\begin{align}
\P\big(\mathbf{x}(w) + \mathbf{z} \in \mathcal{V}_n(\mathbf{x}(0),\mathbf{0},\psi )\big) > 1 - \ec \ . 
\end{align}

Finally, these two bounds can be combined, as in the proof of Lemma \ref{l:achievable}, to get an an achievable red alert exponent\footnote{Note that this is an improvement over the error exponent reported in Theorem 1 of \cite{nd10allerton} since we have used tighter upper bounds. Specifically, in \cite{nd10allerton}, we did not completely take advantage of the fact that both the noise and the standard codewords are nearly orthogonal to any fixed vector and to each other with high probability.} of
\begin{align}
E_{\text{ALERT}} &= \frac{\pp + \pa + 2\sqrt{\pp \pa (1 - e^{-2(C-R)})}}{2N}\nonumber \\ &~~~~ +~ \frac{1}{2}\log\left(\frac{\pa + N}{\pa}\right) - R \ . 
\end{align}

\begin{figure}[h!]
\centering
\includegraphics[width=3.8in]{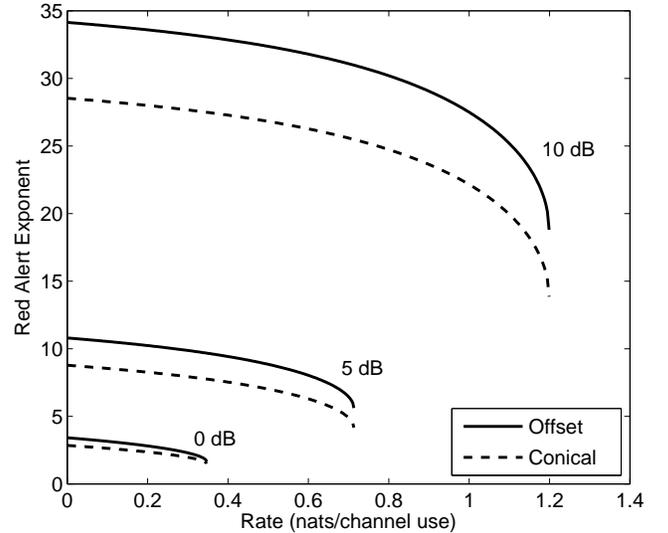}
\caption{Comparison of the red alert exponent attained by an (optimal) offset code construction and a conical code construction with an average power constraint of $0$, $5$, and $10$dB with $\pp = 2 \pa$. }\label{f:conevsoffset}
\end{figure}

In Figure \ref{f:conevsoffset}, we have plotted this red alert exponent alongside the optimal one derived via the offset construction for average power constaints $\pa = 0$, $5,$ and $10$dB with $\pp = 2 \pa$.

\section*{Acknowledgment}

The authors would like to thank the anonymous reviewers whose suggestions improved the presentation of this work. They would also like to thank B. Nakibo\u{g}lu for insightful discussions on the connections between this work and the discrete memoryless case.

\bibliographystyle{ieeetr}

\end{document}